\documentclass[conference, 10pt]{IEEEtran}%

\usepackage{bm}
\usepackage{amssymb}
\usepackage[noadjust]{cite}
\usepackage{graphicx}
\usepackage{subfigure}

\newcommand{\relvar}[2]{\buildrel {#2} \over {#1}}
\newcommand{\eqvar}[1]{\relvar{=}{\mathrm{#1}}}

\newcommand{\levar}[1]{\relvar{\le}{\mathrm{#1}}}

\newcommand{\eqdef}{\eqvar{\scriptscriptstyle\triangle}}

\newcommand{\eqref}[1]{(\ref{#1})}

\newtheorem{theorem}{Theorem}[section]

\newtheorem{proposition}{Proposition}[section]
\newtheorem{corollary}{Corollary}[section]

\begin{document}

\title{Constructing Linear Codes with Good Joint Spectra}
\author{\authorblockN{Shengtian Yang, Yan Chen, Thomas Honold, Zhaoyang Zhang, and Peiliang Qiu}
\authorblockA{Department of Information Science \& Electronic Engineering\\
Zhejiang University\\
Hangzhou, Zhejiang 310027, China\\
\{yangshengtian, qiupl418, honold, ning\_ming, qiupl\}@zju.edu.cn}}

\maketitle

\footnotetext[1]{This work was supported by Zhejiang Provincial Natural Science Foundation of China (No. Y106068), by the National Natural Science Foundation of China (No. 60602023, 60772093), and by the National High Technology Research and Development Program of China (No. 2006AA01Z273, 2007AA01Z257).}

\begin{abstract}
The problem of finding good linear codes for joint source-channel coding (JSCC) is investigated in this paper. By the code-spectrum approach, it has been proved in the authors' previous paper that a good linear code for the authors' JSCC scheme is a code with a good joint spectrum, so the main task in this paper is to construct linear codes with good joint spectra. First, the code-spectrum approach is developed further to facilitate the calculation of spectra. Second, some general principles for constructing good linear codes are presented. Finally, we propose an explicit construction of linear codes with good joint spectra based on low density parity check (LDPC) codes and low density generator matrix (LDGM) codes.
\end{abstract}

\section{Introduction}

A lot of research on practical designs of lossless joint source-channel coding (JSCC) based on linear codes have been done for specific correlated sources and multiple-access channels (MACs), e.g., correlated sources over separated noisy channels (e.g., \cite{MT:Zhao200611a}), correlated sources over additive white Gaussian noise (AWGN) MACs (e.g., \cite{JSCC:Garcia200709}), correlated sources over Rayleigh fading MACs (e.g., \cite{MT:Zhao200611b}). However, for transmission of arbitrary correlated sources over arbitrary MACs, it is still not clear how to construct an optimal lossless JSCC scheme. In \cite{JSCC:Yang200712}, we proposed a lossless JSCC scheme based on linear codes for MACs, which proved to be optimal if good linear codes and good conditional probabilities are chosen. Figure \ref{fig:Scheme1} illustrates the mechanism of our scheme in detail. Using the code-spectrum approach established in \cite{JSCC:Yang200712}, we found that a good linear code for our JSCC scheme is a code with a good joint spectrum. Hence, to design a lossless JSCC scheme in practice, a big problem is how to construct linear codes with good joint spectra. In this paper, we will investigate the problem in depth and give an explicit construction of linear codes with good joint spectra based on sparse matrices.

\begin{figure*}[htbp]
\centering
\includegraphics{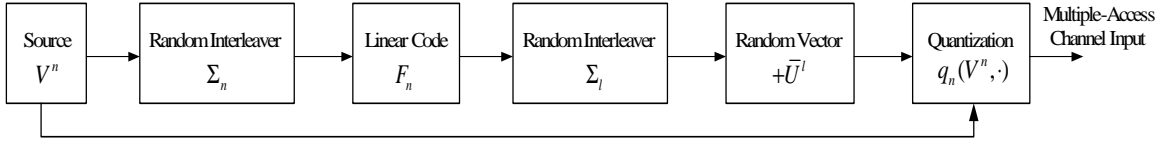}
\caption{The proposed linear codes based lossless joint source-channel encoding scheme of each encoder for multiple-access channels}
\label{fig:Scheme1}
\end{figure*}

In the sequel, symbols, real variables and deterministic mappings are denoted by lowercase letters. Sets and random elements are denoted by capital letters, and the empty set is denoted by $\emptyset$. Alphabets are denoted by script capital letters. All logarithms are taken to the natural base $\mathrm{e}$ and denoted by $\ln$. The composition of the functions $f$ and $g$ is denoted by $f \circ g$, where $(f \circ g)(x) \eqdef f(g(x))$. The indicator function is denoted by $1\{\cdot\}$. The cardinality of a set $A$ is denoted by $|A|$. For any random elements $F$ and $G$ in a common measurable space, the equality $F \eqvar{d} G$ means that $F$ and $G$ have the same probability distribution.

\section{Basics of the Code-Spectrum Approach}\label{sec:CodeSpectrumApproach}

Before investigating the problem of constructing good linear codes, we first need to briefly introduce our ``code-spectrum'' approach established in \cite{JSCC:Yang200712}, which may be regarded as a generalization of the weight-distribution approach (e.g., \cite{JSCC:Divsalar199809}).

Let $\mathcal{X}$ and $\mathcal{Y}$ be two finite (additive) abelian groups. We define a \emph{linear code} as a homomorphism $f: \mathcal{X}^n \to \mathcal{Y}^m$, i.e., a map satisfying
$$
f(x_1^n + x_2^n) = f(x_1^n) + f(x_2^n) \quad \forall x_1^n, x_2^n \in \mathcal{X}^n
$$
where $\mathcal{X}^n$ and $\mathcal{Y}^m$ denote the $n$-fold direct product of $\mathcal{X}$ and the $m$-fold direct product of $\mathcal{Y}$, respectively, and $x^n$ denotes any sequence $x_1x_2\cdots x_n$ in $\mathcal{X}^n$. We also define the rate of a linear code $f$ to be the ratio $n / m$, and denote it by $R(f)$.

Note that any permutation (or interleaver) $\sigma_n$ on $n$ letters can be regarded as an automorphism on $\mathcal{X}^n$. We denote by $\Sigma_n$ a uniformly distributed random permutation on $n$ letters. We tacitly assume that different random permutations occurring in the same expression are independent.

Next, we introduce the concept of types \cite{JSCC:Csiszar198100}. The \emph{type} of a sequence $x^n$ in $\mathcal{X}^n$ is the empirical distribution $P_{x^n}$ on $\mathcal{X}$ defined by
$$
P_{x^n}(a) \eqdef \frac{1}{n} \sum_{i=1}^n 1\{x_i = a\} \quad \forall a \in \mathcal{X}
$$
For a (probability) distribution $P$ on $\mathcal{X}$, the set of sequences of type $P$ in $\mathcal{X}^n$ is denoted by $\mathcal{T}_P^n(\mathcal{X})$. A distribution $P$ on $\mathcal{X}$ is called a type of sequences in $\mathcal{X}^n$ if $\mathcal{T}_P^n(\mathcal{X}) \ne \emptyset$. We denote by $\mathcal{P}(\mathcal{X})$ the set of all distributions on $\mathcal{X}$, and denote by $\mathcal{P}_n(\mathcal{X})$ the set of all possible types of sequences in $\mathcal{X}^n$.

Now, we introduce the \emph{spectrum}, the most important concept in the code-spectrum approach. The spectrum of a nonempty set $A \subseteq \mathcal{X}^n$ is the empirical distribution $S_{\mathcal{X}}(A)$ on $\mathcal{P}(\mathcal{X})$ defined by
$$
S_{\mathcal{X}}(A)(P) \eqdef \frac{|\{x^n \in A | P_{x^n} = P\}|}{|A|} \quad \forall P \in \mathcal{P}(\mathcal{X}).
$$
Analogously, the \emph{joint spectrum} of a nonempty set $B \subseteq \mathcal{X}^n \times \mathcal{Y}^m$ is the empirical distribution $S_{\mathcal{X}\mathcal{Y}}(B)$ on $\mathcal{P}(\mathcal{X}) \times \mathcal{P}(\mathcal{Y})$ defined by
$$
S_{\mathcal{X}\mathcal{Y}}(B)(P, Q) \eqdef \frac{|\{(x^n, y^m) \in B | P_{x^n} = P, P_{y^m} = Q\}|}{|B|}
$$
for all $P \in \mathcal{P}(\mathcal{X}), Q \in \mathcal{P}(\mathcal{Y})$. Furthermore, we define the \emph{marginal spectra} $S_{\mathcal{X}}(B)$, $S_{\mathcal{Y}}(B)$ as the marginal distributions of $S_{\mathcal{X}\mathcal{Y}}(B)$, that is,
$$
S_{\mathcal{X}}(B)(P) \eqdef \sum_{Q \in \mathcal{P}(\mathcal{Y})} S_{\mathcal{X}\mathcal{Y}}(B)(P, Q)
$$
$$
S_{\mathcal{Y}}(B)(Q) \eqdef \sum_{P \in \mathcal{P}(\mathcal{X})} S_{\mathcal{X}\mathcal{Y}}(B)(P, Q).
$$
Please note that the summation in the definition of $S_{\mathcal{X}}(B)(P)$ (or $S_{\mathcal{Y}}(B)(Q)$) is taken over an infinite set $\mathcal{P}(\mathcal{Y})$ (or $\mathcal{P}(\mathcal{X})$), but is actually over a finite set because
$$
S_{\mathcal{X}\mathcal{Y}}(B)(P, Q) = 0
$$
for any $(P, Q)$ satisfying $P \in \mathcal{P}(\mathcal{X}) \backslash \mathcal{P}_n(\mathcal{X})$ or $Q \in \mathcal{P}(\mathcal{Y}) \backslash \mathcal{P}_m(\mathcal{Y})$. We define the \emph{conditional spectra} $S_{\mathcal{Y}|\mathcal{X}}(B)$, $S_{\mathcal{X}|\mathcal{Y}}(B)$ as the conditional distributions of $S_{\mathcal{X}\mathcal{Y}}(B)$, that is,
\begin{IEEEeqnarray*}{l}
S_{\mathcal{Y}|\mathcal{X}}(B)(Q|P) \eqdef \frac{S_{\mathcal{X}\mathcal{Y}}(B)(P, Q)}{S_{\mathcal{X}}(B)(P)} \\
\hspace{10em} \forall P \mbox{ satisfying } S_{\mathcal{X}}(B)(P) \ne 0
\end{IEEEeqnarray*}
\begin{IEEEeqnarray*}{l}
S_{\mathcal{X}|\mathcal{Y}}(B)(P|Q) \eqdef \frac{S_{\mathcal{X}\mathcal{Y}}(B)(P, Q)}{S_{\mathcal{Y}}(B)(Q)} \\
\hspace{10em} \forall Q \mbox{ satisfying } S_{\mathcal{Y}}(B)(Q) \ne 0.
\end{IEEEeqnarray*}

Then naturally, for any given function $f: \mathcal{X}^n \to \mathcal{Y}^m$, we can define its \emph{joint spectrum} $S_{\mathcal{X}\mathcal{Y}}(f)$, \emph{forward conditional spectrum} $S_{\mathcal{Y}|\mathcal{X}}(f)$, \emph{backward conditional spectrum} $S_{\mathcal{X}|\mathcal{Y}}(f)$, and \emph{image spectrum} $S_{\mathcal{Y}}(f)$ as $S_{\mathcal{X}\mathcal{Y}}(\mathrm{rl}(f))$, $S_{\mathcal{Y}|\mathcal{X}}(\mathrm{rl}(f))$, $S_{\mathcal{X}|\mathcal{Y}}(\mathrm{rl}(f))$, and $S_{\mathcal{Y}}(\mathrm{rl}(f))$, respectively, where $\mathrm{rl}(f)$ is the \emph{relation} defined by $\{(x^n, f(x^n)) | x^n \in \mathcal{X}^n\}$. In this case, the forward conditional spectrum is given by
$$
S_{\mathcal{Y}|\mathcal{X}}(f)(Q|P) = \frac{S_{\mathcal{X}\mathcal{Y}}(f)(P, Q)}{S_{\mathcal{X}}(\mathcal{X}^n)(P)}.
$$
If $f$ is a linear code, we further define its \emph{kernel spectrum} as $S_{\mathcal{X}}(\ker f)$, where $\ker\!f \eqdef \{x^n | f(x^n) = 0^m\}$. In this case, we have
$$
S_{\mathcal{Y}}(f) = S_{\mathcal{Y}}(f(\mathcal{X}^n))
$$
since $f$ is a homomorphism according to the definition of linear codes.

The above definitions can be easily extended to more general cases. For example, we may consider the joint spectrum $S_{\mathcal{X}\mathcal{Y}\mathcal{Z}}(C)$ of a set $C \subseteq \mathcal{X}^n \times \mathcal{Y}^m \times \mathcal{Z}^l$, or consider the joint spectrum $S_{\mathcal{X}_1\mathcal{X}_2\mathcal{Y}_1\mathcal{Y}_2}(g)$ of a function $g: \mathcal{X}_1^{n_1} \times \mathcal{X}_2^{n_2} \to \mathcal{Y}_1^{m_1} \times \mathcal{Y}_2^{m_2}$.

A series of properties regarding the spectrum of codes were proved in \cite{JSCC:Yang200712}. Readers may refer to \cite{JSCC:Yang200712} for the details. Some results are listed below for easy reference.

\begin{proposition}\label{pr:SpectrumOfSets}
For all $P \in \mathcal{P}_n(\mathcal{X})$ and $P_i \in \mathcal{P}_{n_i}(\mathcal{X}_i)$ ($1 \le i \le m$), we have
$$
S_{\mathcal{X}}(\mathcal{X}^{n})(P) = \frac{{n \choose nP}}{|\mathcal{X}|^n},
$$
$$
S_{\mathcal{X}_1 \cdots \mathcal{X}_m}(\prod_{i=1}^m A_i)(P_1, \cdots, P_m) = \prod_{i=1}^m S_{\mathcal{X}_i}(A_i)(P_i),
$$
where
$$
{n \choose nP} \eqdef \frac{n!}{\prod_{a \in \mathcal{X}} (nP(a))!}
$$
and $A_i \subseteq \mathcal{X}_i^{n_i}$ ($1 \le i \le m$).
\end{proposition}

\begin{proposition}\label{pr:SpectrumPropertyOfFunctions}
For any given random function $F: \mathcal{X}^n \to \mathcal{Y}^m$, we have
\begin{equation}\label{eq:SpectrumPropertyOfFunctions}
\Pr\{\tilde{F}(x^n) = y^m\} = |\mathcal{Y}|^{-m} \alpha(F)(P_{x^n}, P_{y^m})
\end{equation}
for any $x^n \in \mathcal{X}^n$, $y^m \in \mathcal{Y}^m$, where
\begin{equation}\label{eq:RandomizedF1}
\tilde{F} \eqdef \Sigma_m \circ F \circ \Sigma_n
\end{equation}
and
\begin{IEEEeqnarray}{rCl}
\alpha(F)(P, Q) &\eqdef &\frac{E[S_{\mathcal{X}\mathcal{Y}}(F)(P, Q)]}{S_{\mathcal{X}\mathcal{Y}}(\mathcal{X}^n \times \mathcal{Y}^m)(P, Q)} \IEEEnonumber \\
&= &\frac{E[S_{\mathcal{Y}|\mathcal{X}}(F)(Q|P)]}{S_{\mathcal{Y}}(\mathcal{Y}^m)(Q)}. \label{eq:DefinitionOfAlpha}
\end{IEEEeqnarray}
\end{proposition}

\begin{proposition}\label{pr:GoodLinearCodes}
For any given linear code $f: \mathcal{X}^n \to \mathcal{Y}^m$, we have
\begin{equation}\label{eq:Identity2OfLC}
\alpha(f)(P_{0^n}, P_{y^m}) = |\mathcal{Y}|^{m} 1\{y^m = 0^m\}.
\end{equation}
If both $\mathcal{X}$ and $\mathcal{Y}$ are the Galois field $\mathbb{F}_q$, we define a particular random linear code $F_{q, n, m}^{\mathrm{RLC}}: \mathbb{F}_q^n \to \mathbb{F}_q^m$ by $x^n \mapsto A_{m \times n} \cdot x^n$, where $x^n$ represents an $n$-dimensional column vector, and $A_{m \times n}$ denotes a random matrix with $m$ rows and $n$ columns, each entry independently taking values in $\mathbb{F}_q$ according to a uniform distribution on $\mathbb{F}_q$. (Note that for each realization of $A_{m \times n}$, we then obtain a corresponding realization of $F_{q, n, m}^{\mathrm{RLC}}$. Such a random construction has already been adopted in \cite[Section 2.1]{JSCC:Gallager196300}, \cite{JSCC:Csiszar198207}, etc.) Then we have
\begin{equation}\label{eq:Identity1OfGLC}
\Pr\{\tilde{F}_{q, n, m}^{\mathrm{RLC}}(x^n) = y^m\} = \Pr\{F_{q, n, m}^{\mathrm{RLC}}(x^n) = y^m\} = q^{-m}
\end{equation}
for all $x^n \in \mathbb{F}_q^n \backslash \{0^n\}$ and $y^m \in \mathbb{F}_q^m$, or equivalently
\begin{equation}\label{eq:Identity2OfGLC}
\alpha(F_{q, n, m}^{\mathrm{RLC}})(P, Q) = 1
\end{equation}
for all $P \in \mathcal{P}_n(\mathbb{F}_q) \backslash \{P_{0^n}\}$ and $Q \in \mathcal{P}_m(\mathbb{F}_q)$.
\end{proposition}

\section{Some New Results about Code Spectra}\label{sec:NewResults}

In order to evaluate the performance of a linear code, we need to calculate or estimate its spectrum. However, the results established in \cite{JSCC:Yang200712} are still not enough for this purpose. So in this section, we will present some new results to facilitate the calculation of spectra. All the proofs are easy and hence omitted here.

First, we proved the following two propositions, which imply that any codes (or functions) may be regarded as conditional probability distributions. Such a viewpoint is very helpful when calculating the spectrum of a complex code consisting of many simple codes.

\begin{proposition}\label{pr:SpectrumPropertyXOfFunctions}
For any random function $F: \mathcal{X}^n \to \mathcal{Y}^m$ and any $x^n \in \mathcal{X}^n$, we have
\begin{equation}
\Pr\{(F \circ \Sigma_n)(x^n) \in \mathcal{T}_Q^m(\mathcal{Y})\} = E[S_{\mathcal{Y}|\mathcal{X}}(F)(Q|P_{x^n})].
\end{equation}
\end{proposition}

\begin{proposition}\label{pr:SpectrumOfSeriallyConcatenatedFunctions}
For any two random functions $F: \mathcal{X}^n \to \mathcal{Y}^m$ and $G: \mathcal{Y}^m \to \mathcal{Z}^l$, and any $O \in \mathcal{P}_n(\mathcal{X})$, $Q \in \mathcal{P}_l(\mathcal{Z})$, we have
\begin{IEEEeqnarray*}{Cl}
&E[S_{\mathcal{Z}|\mathcal{X}}(G \circ \Sigma_m \circ F)(Q|O)] \\
= &\sum_{P \in \mathcal{P}_m(\mathcal{Y})} E[S_{\mathcal{Y}|\mathcal{X}}(F)(P|O)] E[S_{\mathcal{Z}|\mathcal{Y}}(G)(Q|P)].
\end{IEEEeqnarray*}
\end{proposition}

Second, let us develop a generating function method for the calculations of spectra.

For any set $A \subseteq \mathcal{X}^n$, we define the \emph{generating function} $\mathcal{G}(A)$ of its spectrum to be
$$
\mathcal{G}(A)(u) \eqdef \sum_{P \in \mathcal{P}_n(\mathcal{X})} S_{\mathcal{X}}(A)(P)(u^{nP})_{\otimes}
$$
where $u$ is a map from $\mathcal{X}$ to $\mathbb{C}$ (the set of complex numbers) or $u \in \mathbb{C}^{\mathcal{X}}$, and for any $u, v \in \mathbb{C}^\mathcal{X}$, we define
$$
(ru)(a) \eqdef ru(a) \quad \forall r \in \mathbb{C}, a \in \mathcal{X},
$$
$$
(u^{v})(a) \eqdef u(a)^{v(a)} \quad \forall a \in \mathcal{X},
$$
$$
(u)_{\otimes} \eqdef \prod_{a \in \mathcal{X}} u(a).
$$
Also note that $P \in \mathcal{P}_n(\mathcal{X}) \subseteq \mathbb{C}^{\mathcal{X}}$. Analogously, for any set $B \subseteq \mathcal{X}^n \times \mathcal{Y}^m$, we define the generating function of its joint spectrum as
\begin{IEEEeqnarray*}{l}
\mathcal{G}(B)(u, v) \eqdef \\
\quad \sum_{P \in \mathcal{P}_n(\mathcal{X}), Q \in \mathcal{P}_m(\mathcal{Y})} S_{\mathcal{X}\mathcal{Y}}(B)(P, Q) (u^{nP})_{\otimes} (v^{mQ})_{\otimes},
\end{IEEEeqnarray*}
where $u \in \mathbb{C}^{\mathcal{X}}$, $v \in \mathbb{C}^{\mathcal{Y}}$. This in particular defines $\mathcal{G}(f) \eqdef \mathcal{G}(\mathrm{rl}(f))$ for any function $f: \mathcal{X}^n \to \mathcal{Y}^m$.

Based on the above definitions, we proved the following properties.

\begin{proposition}\label{pr:GeneratingFunctionProperty}
For any two sets $A_1 \subseteq \mathcal{X}^{n_1}$ and $A_2 \subseteq \mathcal{X}^{n_2}$, we have
\begin{equation}\label{eq:GeneratingFunctionProperty2A}
\mathcal{G}(A_1 \times A_2)(u) = \mathcal{G}(A_1)(u) \cdot \mathcal{G}(A_2)(u).
\end{equation}
For any two sets $A_1 \subseteq \mathcal{X}^{n}$ and $A_2 \subseteq \mathcal{Y}^{m}$, we have
\begin{equation}\label{eq:GeneratingFunctionProperty2B}
\mathcal{G}(A_1 \times A_2)(u, v) = \mathcal{G}(A_1)(u) \cdot \mathcal{G}(A_2)(v).
\end{equation}
For any two sets $B_1 \subseteq \mathcal{X}^{n_1} \times \mathcal{Y}^{m_1}$ and $B_2 \subseteq \mathcal{X}^{n_2} \times \mathcal{Y}^{m_2}$, we have
\begin{equation}\label{eq:GeneratingFunctionProperty2C}
\mathcal{G}(B_1 \times B_2)(u, v) = \mathcal{G}(B_1)(u, v) \cdot \mathcal{G}(B_2)(u, v).
\end{equation}
\end{proposition}

\begin{corollary}
$$
\mathcal{G}(\mathcal{X}^n)(u) = \biggl( \frac{(u)_{\oplus}}{|\mathcal{X}|} \biggr)^n,
$$
where
$$
(u)_{\oplus} \eqdef \sum_{a \in \mathcal{X}} u(a).
$$
\end{corollary}

\begin{corollary}\label{co:GeneratingFunctionProperty}
For any two functions $f_1: \mathcal{X}^{n_1} \to \mathcal{Y}^{m_1}$ and $f_2: \mathcal{X}^{n_2} \to \mathcal{Y}^{m_2}$, we have
\begin{equation}
\mathcal{G}(f_1 \odot f_2)(u, v) = \mathcal{G}(f_1)(u, v) \cdot \mathcal{G}(f_2)(u, v),
\end{equation}
where $f_1 \odot f_2$ is the map from $\mathcal{X}^{n_1+n_2}$ to $\mathcal{Y}^{m_1+m_2}$ defined by
$$
(f_1 \odot f_2)(x^{n_1+n_2}) \eqdef f_1(x_{1 \cdots n_1}) f_2(x_{(n_1+1) \cdots (n_1+n_2)})
$$
for all $x^{n_1+n_2} \in \mathcal{X}^{n_1+n_2}$.
\end{corollary}

\section{General Principles for Constructing Linear Codes with Good Joint Spectra}\label{sec:IV}

In this section, we will investigate some general problems for constructing linear codes with good joint spectra.

At first, we need to introduce some concepts of good linear codes. According to \cite[Table I]{JSCC:Yang200712}, a sequence of $\delta$-asymptotically good (random) linear codes $F_n: \mathcal{X}^n \to \mathcal{Y}^{m_n}$ for JSCC is a sequence of linear codes whose joint spectra satisfy
$$
\limsup_{n \to \infty} \max_{\scriptstyle P \in \mathcal{P}_n(\mathcal{X}) \backslash \{P_{0^n}\}, \atop \scriptstyle Q \in \mathcal{P}_{m_n}(\mathcal{Y})} \frac{1}{n} \ln \frac{E[S_{\mathcal{X}\mathcal{Y}}(F_n)(P,Q)]}{S_{\mathcal{X}\mathcal{Y}}(\mathcal{X}^n \times \mathcal{Y}^{m_n})(P,Q)} \le \delta.
$$
And for comparison, a sequence of $\delta$-asymptotically good (random) linear codes $F_n: \mathcal{X}^n \to \mathcal{Y}^{m_n}$ for channel coding is a sequence of linear codes whose image spectra satisfy
$$
\limsup_{n \to \infty} \max_{Q \in \mathcal{P}_{m_n}(\mathcal{Y}) \backslash \{P_{0^{m_n}}\}} \frac{1}{m_n} \ln \frac{E[S_{\mathcal{Y}}(F_n(\mathcal{X}^n))(Q)]}{S_{\mathcal{Y}}(\mathcal{Y}^{m_n})(Q)} \le \delta.
$$
When $\delta$ equals zero, the above codes are then called asymptotically good linear codes for JSCC and channel coding, respectively.

From the linearity of codes, it follows easily that $\delta$-asymptotically good linear codes for JSCC are subsets of $\delta\overline{R}(\bm{F})$-asymptotically good linear codes for channel coding, where $\overline{R}(\bm{F}) \eqdef \limsup_{n \to \infty} R(F_n)$ and $\bm{F} \eqdef \{F_n\}_{n=1}^\infty$. Then naturally, our \emph{first problem} is: if a sequence of asymptotically good linear codes $f_n$ for channel coding is given, can we find a sequence of asymptotically good linear codes $g_n$ for JSCC such that $g_n(\mathcal{X}^n) = f_n(\mathcal{X}^n)$? In other words (assuming that $\mathcal{X} = \mathcal{Y} = \mathbb{F}_q$), if a sequence of asymptotically good channel codes is given, can we choose a good sequence of generator matrices so that the linear codes are asymptotically good for JSCC?

When $\mathcal{X} = \mathbb{F}_q$, the answer is positive, as a consequence of the following theorem.

\begin{theorem}\label{th:GoodJointSpectrumFromGoodImageSpectrum}
For any linear code $f: \mathbb{F}_q^n \to \mathcal{Y}^m$, there exists a linear code $g: \mathbb{F}_q^n \to \mathcal{Y}^m$ such that
$$
g(\mathbb{F}_q^n) = f(\mathbb{F}_q^n)
$$
and
\begin{equation}
S_{\mathcal{Y}|\mathbb{F}_q}(g)(Q|P) < \frac{S_{\mathcal{Y}}(f(\mathbb{F}_q^n))(Q)}{1 - q^{-1} - q^{-2}}
\end{equation}
for all $P \in \mathcal{P}_n(\mathbb{F}_q) \backslash \{P_{0^n}\}$, $Q \in \mathcal{P}_m(\mathcal{Y})$.
\end{theorem}

\begin{proof}[Sketch of Proof]
The main idea of the proof is to construct a random linear code $G \eqdef f \circ F_{q,n,n}^{\mathrm{RLC}} \eqvar{d} f \circ \Sigma_n \circ F_{q,n,n}^{\mathrm{RLC}}$, where $F_{q,n,n}^{\mathrm{RLC}}$ is defined in Proposition \ref{pr:GoodLinearCodes}. Then by Proposition \ref{pr:SpectrumOfSeriallyConcatenatedFunctions}, we have $E[S_{\mathcal{Y}|\mathbb{F}_q}(G)(Q|P)] = S_{\mathcal{Y}}(f(\mathbb{F}_q^n))(Q)$ for all $P \ne P_{0^n}$ and $Q$. This together with Proposition \ref{pr:RankOfRLC} (see below) then concludes the theorem.
\end{proof}

\begin{proposition}\label{pr:RankOfRLC}
Let $\mathrm{rank}(F)$ be the rank of the generator matrix of the linear code $F: \mathbb{F}_q^n \to \mathbb{F}_q^m$. Then we have
\begin{equation}\label{eq:RankOfRLC1}
\Pr\{\mathrm{rank}(F_{q,n,m}^{\mathrm{RLC}}) = m\} = \prod_{i=1}^m \biggl( 1 - \frac{q^{i-1}}{q^n} \biggr)
\end{equation}
where $F_{q,n,m}^{\mathrm{RLC}}$ is defined in Proposition \ref{pr:GoodLinearCodes} and $m \le n$. Furthermore, we have
\begin{equation}\label{eq:RankOfRLC2}
\prod_{i=1}^{m} \biggl( 1 - \frac{q^{i-1}}{q^n} \biggr) > \biggl( 1 - \frac{q^{m-n-k}}{q - 1} \biggr) \prod_{i=1}^{k} (1 - q^{m-n-i}),
\end{equation}
where $1 \le k \le m$. Let $m = n$ and $k = 1$, then we have
\begin{equation}\label{eq:RankOfRLC3}
\Pr\{\mathrm{rank}(F_{q,n,n}^{\mathrm{RLC}}) = n\} > 1 - q^{-1} - q^{-2}.
\end{equation}
\end{proposition}

\begin{proof}
The identity \eqref{eq:RankOfRLC1} is a well known result in probability theory. To obtain a lower bound of the right hand side of \eqref{eq:RankOfRLC1}, we have
\begin{IEEEeqnarray*}{rCl}
\prod_{i=1}^{m} \biggl( 1 - \frac{q^{i-1}}{q^n} \biggr)
&= &\prod_{i=1}^{m-k} \biggl( 1 - \frac{q^{i-1}}{q^n} \biggr) \prod_{i=1}^{k} \biggl( 1 - \frac{q^{m-i}}{q^n} \biggr) \\
&\ge &\biggl( 1 - \sum_{i=1}^{m-k} \frac{q^{i-1}}{q^n} \biggr) \prod_{i=1}^{k} \biggl( 1 - \frac{q^{m-i}}{q^n} \biggr) \\
&= &\biggl( 1 - \frac{q^{m-k} - 1}{q^n(q - 1)} \biggr) \prod_{i=1}^{k} \biggl( 1 - \frac{q^{m-i}}{q^n} \biggr) \\
&> &\biggl( 1 - \frac{q^{m-n-k}}{q - 1} \biggr) \prod_{i=1}^{k} \biggl( 1 - \frac{q^{m-i}}{q^n} \biggr).
\end{IEEEeqnarray*}
This concludes \eqref{eq:RankOfRLC2}, and \eqref{eq:RankOfRLC3} follows clearly.
\end{proof}

The above result does give a possible way for constructing good linear codes for JSCC based on good channel codes. However, such a construction is somewhat difficult to implement in practice, because the random generator matrix of $F_{q,n,m}^{\mathrm{RLC}}$ is dense. Thus, our \emph{second problem} is how to construct linear codes with good joint spectra based on sparse matrices so that known iterative encoding and decoding procedures have low complexity. The following theorem gives one feasible solution.

\begin{theorem}\label{th:ConstructionOfGoodLineraCodes}
For a given sequence of sets $\{A_n \subseteq \mathcal{P}_{m_n}(\mathcal{X}) \backslash \{P_{0^{m_n}}\}\}_{n=1}^\infty$, if there exist two sequences of random linear codes $F_n: \mathcal{X}^n \to \mathcal{X}^{m_n}$ and $G_n: \mathcal{X}^{m_n} \to \mathcal{X}^{l_n}$ satisfying
\begin{equation}\label{eq:ConstructionOfGoodLineraCodes1}
F_n(\mathcal{X}^{n} \backslash \{0^n\}) \subseteq \bigcup_{P \in A_n} \mathcal{T}_{P}^{m_n}(\mathcal{X})
\end{equation}
and
\begin{equation}\label{eq:ConstructionOfGoodLineraCodes2}
\limsup_{n \to \infty} \max_{P \in A_n, Q \in \mathcal{P}_{l_n}(\mathcal{X})} \frac{1}{n} \ln \frac{E[S_{\mathcal{X}|\mathcal{X}}(G_n)(Q|P)]}{S_{\mathcal{X}}(\mathcal{X}^{l_n})(Q)} \le \delta
\end{equation}
respectively, where $\delta \ge 0$, then we have
\begin{IEEEeqnarray*}{l}
\limsup_{n \to \infty} \max_{O \in \mathcal{P}_{n}(\mathcal{X}) \backslash \{P_{0^{n}}\}, Q \in \mathcal{P}_{l_n}(\mathcal{X})} \\
\qquad \frac{1}{n} \ln \frac{E[S_{\mathcal{X}|\mathcal{X}}(G_n \circ \Sigma_{m_n} \circ F_n)(Q|O)]}{S_{\mathcal{X}}(\mathcal{X}^{l_n})(Q)} \le \delta.
\end{IEEEeqnarray*}
\end{theorem}

\begin{proof}
For all $O \in \mathcal{P}_{n}(\mathcal{X}) \backslash \{P_{0^{n}}\}$ and $Q \in P_{l_n}(\mathcal{X})$, and for any $\epsilon > 0$ and sufficiently large $n$, we have
\begin{IEEEeqnarray*}{Cl}
&S_{\mathcal{X}|\mathcal{X}}(G_n \circ \Sigma_{m_n} \circ F_n)(Q|O) \\
\eqvar{(a)} &\sum_{P \in \mathcal{P}_{m_n}(\mathcal{X})} E[S_{\mathcal{X}|\mathcal{X}}(G_n)(Q|P)] E[S_{\mathcal{X}|\mathcal{X}}(F_n)(P|O)] \\
\eqvar{(b)} &\sum_{P \in A_n} E[S_{\mathcal{X}|\mathcal{X}}(G_n)(Q|P)] E[S_{\mathcal{X}|\mathcal{X}}(F_n)(P|O)] \\
\levar{(c)} &\sum_{P \in A_n} e^{n(\delta + \epsilon)} S_{\mathcal{X}}(\mathcal{X}^{l_n})(Q) E[S_{\mathcal{X}|\mathcal{X}}(F_n)(P|O)] \\
\le &e^{n(\delta + \epsilon)} S_{\mathcal{X}}(\mathcal{X}^{l_n})(Q),
\end{IEEEeqnarray*}
where (a) follows from Proposition \ref{pr:SpectrumOfSeriallyConcatenatedFunctions}, (b) from the condition \eqref{eq:ConstructionOfGoodLineraCodes1}, and (c) from the condition \eqref{eq:ConstructionOfGoodLineraCodes2}.
Therefore, for any $\epsilon > 0$ and sufficiently large $n$,
\begin{IEEEeqnarray*}{l}
\max_{O \in \mathcal{P}_{n}(\mathcal{X}) \backslash \{P_{0^{n}}\}, Q \in \mathcal{P}_{l_n}(\mathcal{X})} \\
\qquad \frac{1}{n} \ln \frac{E[S_{\mathcal{X}|\mathcal{X}}(G_n \circ \Sigma_{m_n} \circ F_n)(Q|O)]}{S_{\mathcal{X}}(\mathcal{X}^{l_n})(Q)} \le \delta + \epsilon,
\end{IEEEeqnarray*}
which establishes the theorem.
\end{proof}

Using Theorem \ref{th:ConstructionOfGoodLineraCodes}, we can now construct good linear codes by a serial concatenation scheme, where the inner code is approximately $\delta$-asymptotically good (satisfying \eqref{eq:ConstructionOfGoodLineraCodes2}) and the outer code is a linear code having good distance properties if we set $A_n = \{P \in \mathcal{P}_{m_n}(\mathcal{X}) | 1 - P(0) > \gamma\}$ in the condition \eqref{eq:ConstructionOfGoodLineraCodes1}. According to \cite[Section IV]{JSCC:Bennatan200403}, there exists a good low density parity check (LDPC) code $F_n$ satisfying \eqref{eq:ConstructionOfGoodLineraCodes1} for an appropriate $\gamma$. Then our \emph{final problem} is how to find a sequence of approximately $\delta$-asymptotically good linear codes satisfying \eqref{eq:ConstructionOfGoodLineraCodes2} with $A_n = \{P \in \mathcal{P}_{m_n}(\mathcal{X}) | 1 - P(0) > \gamma\}$. In the next section, we will find such codes in a family of codes called low density generator matrix (LDGM) codes.

\section{The spectra of LDGM Codes}

In this section, we will investigate the joint spectra of LDGM codes. We assume that the alphabet of codes is $\mathbb{F}_q$, and we denote a regular LDGM code by the map $F_{n, c, d}^{\mathrm{LD}}: \mathbb{F}_q^n \to \mathbb{F}_q^m$ defined by
$$
F_{n, c, d}^{\mathrm{LD}} \eqdef (\odot_{i=1}^m F_d^{\mathrm{CHK}}) \circ \Sigma_{cn} \circ (\odot_{i=1}^n f_c^{\mathrm{REP}})
$$
where $nc = md$, and $f_c^{\mathrm{REP}}$ is a single symbol repetition code $f_c^{\mathrm{REP}}: \mathbb{F}_q \to \mathbb{F}_q^c$ defined by
$$
f_c^{\mathrm{REP}}(x) \eqdef x x \cdots x \quad \forall x \in \mathbb{F}_q,
$$
and $\odot_{i=1}^m F_d^{\mathrm{CHK}}$ denotes a parallel concatenation of $m$ independent copies of the random single symbol check code $F_d^{\mathrm{CHK}}: \mathbb{F}_q^d \to \mathbb{F}_q$ defined by
$$
F_d^{\mathrm{CHK}}(x^d) \eqdef \sum_{i=1}^d C_i x_i \quad \forall x^d \in \mathbb{F}_q^d
$$
where $C_i$ ($i = 1, 2, \cdots, d$) denotes an independent uniform random variable on the set $\mathbb{F}_q \backslash \{0\}$.

To evaluate the joint spectrum or conditional spectrum of $F_{n, c, d}^{\mathrm{LD}}$, we first need to calculate the joint spectrum or conditional spectrum of $f_c^{\mathrm{REP}}$ and $F_d^{\mathrm{CHK}}$. We have the following results.

\begin{proposition}\label{pr:SpectrumOfREPCode}
\begin{equation}\label{eq:SpectrumOfREPCodeA}
\mathcal{G}(f_c^{\mathrm{REP}})(u, v) = \frac{1}{q} \sum_{a \in \mathbb{F}_q} u(a)[v(a)]^c,
\end{equation}
\begin{IEEEeqnarray}{l}
\mathcal{G}(\odot_{i=1}^n f_c^{\mathrm{REP}})(u, v) \IEEEnonumber \\
\qquad = \frac{1}{q^n} \sum_{P \in \mathcal{P}_n(\mathbb{F}_q)} {n \choose nP} (u^{nP})_{\otimes} (v^{ncP})_{\otimes},\label{eq:SpectrumOfREPCodeB}
\end{IEEEeqnarray}
\begin{equation}\label{eq:SpectrumOfREPCodeC}
S_{\mathbb{F}_q\mathbb{F}_q}(\odot_{i=1}^n f_c^{\mathrm{REP}})(P, Q) = S_{\mathbb{F}_q}(\mathbb{F}_q^n)(P) 1\{P = Q\},
\end{equation}
\begin{equation}\label{eq:SpectrumOfREPCodeD}
S_{\mathbb{F}_q|\mathbb{F}_q}(\odot_{i=1}^n f_c^{\mathrm{REP}})(Q | P) = 1\{P = Q\}.
\end{equation}
\end{proposition}

\begin{proof}
The identity \eqref{eq:SpectrumOfREPCodeA} holds clearly. From \eqref{eq:SpectrumOfREPCodeA} and Corollary \ref{co:GeneratingFunctionProperty}, we then have
\begin{IEEEeqnarray*}{Cl}
&\mathcal{G}(\odot_{i=1}^n f_c^{\mathrm{REP}})(u, v) \\
= &\biggl( \frac{1}{q} \sum_{a \in \mathbb{F}_q} u(a)[v(a)]^c \biggr)^n \\
= &\frac{1}{q^n} \sum_{P \in \mathcal{P}_n(\mathbb{F}_q)} {n \choose nP} \prod_{a \in \mathbb{F}_q} [u(a)]^{nP(a)} [v(a)]^{ncP(a)} \\
= &\frac{1}{q^n} \sum_{P \in \mathcal{P}_n(\mathbb{F}_q)} {n \choose nP} (u^{nP})_{\otimes} (v^{ncP})_{\otimes}.
\end{IEEEeqnarray*}
This proves \eqref{eq:SpectrumOfREPCodeB}. The identities \eqref{eq:SpectrumOfREPCodeC} and \eqref{eq:SpectrumOfREPCodeD} are easy consequences of \eqref{eq:SpectrumOfREPCodeB}.
\end{proof}

In order to obtain the joint spectrum of $F_d^{\mathrm{CHK}}$, we need the following proposition (also well known), which can be easily proved by mathematical induction.

\begin{proposition}\label{pr:CheckSum}
Let
\begin{equation}\label{eq:CheckSumA}
Y_d = \sum_{i=1}^d X_i,
\end{equation}
where $X_i$ ($i = 1, 2, \cdots, d$) is an independent uniform random variable on the set $\mathbb{F}_q \backslash \{0\}$. Then we have
\begin{IEEEeqnarray}{rCl}
\Pr\{Y_d = a\} &= &1\{a = 0\} \frac{1}{q} \biggl[ 1 - \bigl( -\frac{1}{q-1} \bigl)^{d-1} \biggr] + \IEEEnonumber \\
& &1\{a \ne 0\} \frac{1}{q} \biggl[ 1 - \bigl( -\frac{1}{q-1} \bigl)^{d} \biggr]\label{eq:CheckSumB}
\end{IEEEeqnarray}
for any $a \in \mathrm{GF}(q)$.
\end{proposition}

Now let us calculate the joint spectrum of $F_d^{\mathrm{CHK}}$. By Proposition \ref{pr:SpectrumPropertyOfFunctions}, \ref{pr:CheckSum} and Corollary \ref{co:GeneratingFunctionProperty}, we obtained the following proposition. Its proof is long and hence omitted here, and readers may refer to \cite{JSCC:Yang200800} for the details.

\begin{proposition}\label{pr:SpectrumOfCHKCode}
\begin{IEEEeqnarray}{l}
E[\mathcal{G}(F_d^{\mathrm{CHK}})(u, v)] = \frac{1}{q^{d+1}} \biggl[ ((u)_\oplus)^d (v)_\oplus + \IEEEnonumber \\
\qquad \biggl( \frac{qu(0) - (u)_\oplus}{q-1} \biggr)^d (qv(0) - (v)_\oplus) \biggr],\label{eq:SpectrumOfCHKCodeA}
\end{IEEEeqnarray}
\begin{equation}\label{eq:SpectrumOfCHKCodeB}
E[S_{\mathbb{F}_q\mathbb{F}_q}(\odot_{i=1}^m F_d^{\mathrm{CHK}})(P, Q)] = \mathrm{coef}(g_1(u,Q), (u^{mdP})_\otimes),
\end{equation}
\begin{IEEEeqnarray}{l}
E[S_{\mathbb{F}_q\mathbb{F}_q}(\odot_{i=1}^m F_d^{\mathrm{CHK}})(P, Q)] \le g_2(O,P,Q), \IEEEnonumber \\
\qquad \forall O \in \mathcal{P}_{md}(\mathbb{F}_q) \; (O(a) > 0, \forall a \in \{a|P(a) > 0\}),\label{eq:SpectrumOfCHKCodeC}
\end{IEEEeqnarray}
\begin{equation}\label{eq:SpectrumOfCHKCodeD}
\frac{1}{m} \ln \alpha(\odot_{i=1}^m F_d^{\mathrm{CHK}})(P, Q) \le \delta_d(P(0),Q(0)) + \mathrm{O}\biggr(\frac{\ln m}{m}\biggr),
\end{equation}
where $\mathrm{coef}(f(u), (u^{v})_\otimes)$ denotes the coefficient of $(u^{v})_\otimes$ in the polynomial $f(u)$, and
\begin{IEEEeqnarray*}{Cl}
&g_1(u, Q) \\
\eqdef &\frac{{m \choose mQ}}{q^{m(d+1)}} \biggl[ ((u)_\oplus)^d + (q - 1) \biggl( \frac{qu(0) - (u)_\oplus}{q-1} \biggr)^d \biggr]^{mQ(0)} \\
&\biggl[ ((u)_\oplus)^d - \biggl( \frac{qu(0) - (u)_\oplus}{q-1} \biggr)^d \biggr]^{m(1-Q(0))},
\end{IEEEeqnarray*}
\begin{IEEEeqnarray*}{Cl}
&g_2(O, P, Q) \\
\eqdef &\frac{{m \choose mQ}}{q^{m(d+1)} (O^{mdP})_\otimes} \biggl[ 1 + (q - 1) \biggl( \frac{qO(0) - 1}{q-1} \biggr)^d \biggr]^{mQ(0)} \\
&\biggl[ 1 - \biggl( \frac{qO(0) - 1}{q-1} \biggr)^d \biggr]^{m(1-Q(0))},
\end{IEEEeqnarray*}
\begin{IEEEeqnarray}{rCl}
\delta_d(x, y) &\eqdef &\inf_{0 < \hat{x} < 1} \biggl\{ dD(x \| \hat{x}) + y \ln \biggl[ 1 + (q - 1) \biggl( \frac{q\hat{x} - 1}{q-1} \biggr)^d \biggr] \IEEEnonumber \\
& &+\: (1 - y) \ln \biggl[ 1 - \biggl( \frac{q\hat{x} - 1}{q-1} \biggr)^d \biggr] \biggr\}.\label{eq:DefinitionOfDeltad}
\end{IEEEeqnarray}
where $D(x \| \hat{x})$ is the \emph{information divergence} defined by
$$
D(x \| \hat{x}) \eqdef x \ln \frac{x}{\hat{x}} + (1 - x) \ln \frac{1 - x}{1 - \hat{x}}.
$$
\end{proposition}

Based on the above preparations, we now start to analyze the joint spectrum of the regular LDGM code $F_{n,c,d}^{\mathrm{LD}}$.

\begin{theorem}\label{th:SpectrumOfLDCode}
\begin{equation}\label{eq:SpectrumOfLDCodeA}
\frac{1}{n} \ln \alpha(F_{n, c, d}^{\mathrm{LD}})(P, Q) \le \frac{c}{d} \delta_d(P(0), Q(0)) + \mathrm{O}\biggl( \frac{\ln n}{n} \biggr).
\end{equation}
where $\delta_d$ is defined by \eqref{eq:DefinitionOfDeltad}. Let
\begin{equation}\label{eq:SpectrumOfLDCodeB}
A_n(\gamma) \eqdef \{P \in \mathcal{P}_n(\mathbb{F}_q)| 1 - P(0) > \gamma\},
\end{equation}
where $0 < \gamma < 1$. Then, when $q > 2$, for any $0 < \gamma < 1$ and any $\delta > 0$, there exits a positive integer $d_0 = d_0(\gamma, \delta)$ such that
\begin{equation}\label{eq:SpectrumOfLDCodeC}
\limsup_{n \to \infty} \max_{P \in A_n(\gamma), Q \in \mathcal{P}_{m}(\mathbb{F}_q)} \frac{1}{n} \ln \alpha(F_{n, c, d}^{\mathrm{LD}})(P, Q) \le \delta
\end{equation}
for all integers $d \ge d_0$.
\end{theorem}

\begin{proof}
At first, according to the definition of regular LDGM codes, we have
\begin{IEEEeqnarray*}{Cl}
&E[S_{\mathbb{F}_q|\mathbb{F}_q}(F_{n, c, d}^{\mathrm{LD}})(Q|P)] \\
\eqvar{(a)} &\sum_{O \in \mathcal{P}_{nc}(\mathbb{F}_q)} E[S_{\mathbb{F}_q|\mathbb{F}_q}(\odot_{i=1}^n f_c^{\mathrm{REP}})(O|P)] \:\cdot\\
&E[S_{\mathbb{F}_q|\mathbb{F}_q}(\odot_{i=1}^m F_d^{\mathrm{CHK}})(Q|O)] \\
\eqvar{(b)} &\sum_{O \in \mathcal{P}_{nc}(\mathbb{F}_q)} 1\{P = O\} E[S_{\mathbb{F}_q|\mathbb{F}_q}(\odot_{i=1}^m F_d^{\mathrm{CHK}})(Q|O)] \\
= &E[S_{\mathbb{F}_q|\mathbb{F}_q}(\odot_{i=1}^m F_d^{\mathrm{CHK}})(Q|P)] \\
\levar{(c)} &e^{m(\delta_d(P(0), Q(0)) + \mathrm{O}(\frac{\ln m}{m}))} S_{\mathbb{F}_q}(\mathbb{F}_q^m)(Q),
\end{IEEEeqnarray*}
where (a) follows from  Proposition \ref{pr:SpectrumOfSeriallyConcatenatedFunctions} and the definition of $F_{n,c,d}^{\mathrm{LD}}$, (b) from Proposition \ref{pr:SpectrumOfREPCode}, and (c) from Proposition \ref{pr:SpectrumOfCHKCode}. This then concludes \eqref{eq:SpectrumOfLDCodeA}.

By the definition of $\delta_d$, we have
\begin{IEEEeqnarray*}{rCl}
\delta_d(x, y)
&\le &dD(x \| x) + y \ln \biggl[ 1 + (q - 1) \biggl( \frac{qx - 1}{q-1} \biggr)^d \biggr] \\
& &+\: (1 - y) \ln \biggl[ 1 - \biggl( \frac{qx - 1}{q-1} \biggr)^d \biggr] \\
&\le &\ln \biggl[ 1 + (qy - 1) \biggl( \frac{qx - 1}{q-1} \biggr)^d \biggr] \\
&\le &(qy - 1) \biggl( \frac{qx - 1}{q-1} \biggr)^d \\
&\le &(q - 1) \left| \frac{qx - 1}{q-1} \right|^d.
\end{IEEEeqnarray*}
Furthermore, when $0 \le x < 1 - \gamma$ and $q > 2$, we have
$$
-1 < - \frac{1}{q-1} \le \frac{qx - 1}{q - 1} \le 1 - \frac{q\gamma}{q - 1} < 1
$$
or
$$
\left| \frac{qx - 1}{q - 1} \right| < \max\{\frac{1}{q-1}, 1 - \frac{q\gamma}{q - 1}\} < 1.
$$
Then there exists a positive integer $d_0 = d_0(\gamma, \delta)$ such that
$$
\sup_{\scriptstyle 0 \le x < 1 - \gamma, \atop \scriptstyle 0 \le y \le 1} \delta_d(x, y) \le \frac{d \delta}{c} \quad \forall d \ge d_0.
$$
Note here that the ratio $d/c$ is the rate of the code and hence should be a constant or at least bounded.

Therefore, for any $d \ge d_0$, we have
\begin{IEEEeqnarray*}{Cl}
&\max_{P \in A_n(\gamma), Q \in \mathcal{P}_{m}(\mathbb{F}_q)} \frac{1}{n} \ln \alpha(F_{n, c, d}^{\mathrm{LD}})(P, Q) \\
\levar{(a)} &\frac{c}{d} \sup_{\scriptstyle 0 \le x < 1 - \gamma, \atop \scriptstyle 0 \le y \le 1} \delta_d(x, y) + \mathrm{O}\biggl( \frac{\ln n}{n} \biggr) \\
\le &\delta  + \mathrm{O}\biggl( \frac{\ln n}{n} \biggr),
\end{IEEEeqnarray*}
where (a) follows from \eqref{eq:SpectrumOfLDCodeA}. This concludes \eqref{eq:SpectrumOfLDCodeC}.
\end{proof}

Theorem \ref{th:SpectrumOfLDCode} actually exhibits a family of codes whose joint spectra are approximately $\delta$-asymptotically good. Then together with the conclusion at the end of Section \ref{sec:IV}, we have completed the construction of linear codes with good joint spectra, i.e., a serial concatenation scheme with one LDPC code as an outer code and one LDGM code as an inner code. An analogous construction has been proposed by Hsu in his thesis \cite{JSCC:Hsu200600}, but his purpose was only to find good channel codes and only a rate-1 LDGM code was employed as an inner code in his construction.

\bibliographystyle{IEEEtran} % use IEEEtran.bst style
\bibliography{IEEEabrv,JSCC4}

\end{document}